\documentclass[twocolumn]{autart}    

\usepackage{psfrag}
\usepackage{pstricks}
\usepackage{pstricks-add}
\usepackage{graphics} 
\usepackage{graphicx}          
\usepackage{ulem}
\usepackage{epsfig} 
\usepackage{mathptmx} 
\usepackage{times} 
\usepackage{amsmath} 
\usepackage{amssymb}  
\usepackage{amsfonts}
\usepackage{pgfplots}

\usepackage{cite}
\usepackage{float} 

\pgfplotsset{compat=1.13}

\newtheorem {remark}{Remark}
\newtheorem {corollary}{Corollary}
\newtheorem {proposition}{Proposition}
\newtheorem {property}{Property}
\newtheorem {definition}{Definition}

\def\R{\mathbb{R}}
\def\B{\mathbb{B}}
\def\N{\mathbb{N}}
\def\Z{\mathbb{Z}}

\def\S{\mathbb{S}}

\def\E#1{{\rm{E}}\{#1\}}
\def\Pr{\mathrm{Pr}}
\def\tr{\mathrm{tr}}

\def\Gij{\Gamma_{i,j}}
\def\Gik{\Gamma_{i,k}}
\def\Gma{\tilde{\Gamma}}
\def\bx{\bar{x}}
\def\cR{\mathcal{R}}

\newcommand {\T}{^{\top}} 

\everymath{\displaystyle}

\begin{document}

\begin{frontmatter}
\title{Probabilistic reachable and invariant sets \\ 
for linear systems with correlated disturbance} 

\thanks[footnoteinfo]{This paper was not presented at any IFAC 
meeting. Corresponding author.}

\author[GIPSA]{Mirko Fiacchini}\ead{mirko.fiacchini@gipsa-lab.fr},
\author[Sevilla]{Teodoro Alamo}\ead{teodoroalamo@gmail.com},

\address[GIPSA]{Univ. Grenoble Alpes, CNRS, Grenoble INP, GIPSA-lab, 38000 
Grenoble, France}
\address[Sevilla]{Departamento de Ingenier\'ia de Sistemas y Autom\'atica, 
Universidad de Sevilla, Sevilla 41092, Spain}


\begin{abstract} In this paper a constructive method to determine and compute 
probabilistic reachable and invariant sets for linear discrete-time systems, 
excited by a stochastic disturbance, is presented. The samples of the 
disturbance signal are not assumed to be uncorrelated, only a bound on the 
correlation matrices is supposed to be known.  The concept of correlation 
bound is introduced and employed to determine probabilistic reachable sets and 
probabilistic invariant sets. Constructive methods for their computation, based 
on convex optimization, are given. 
\end{abstract}

\end{frontmatter}

\section{Introduction}

The recent interest in the characterization and computation of probabilistic 
reachable sets and probabilistic invariant sets is mostly due to the growing 
popularity of stochastic Model Predictive Control (SMPC), see 
\cite{mesbah2016stochastic}. Indeed, as for deterministic and robust predictive 
techniques, several desirable features can be ensured also in the stochastic 
context by appropriately employing reachable and invariant sets to ensure 
probabilistic guarantees, for instance, of constraints satisfaction, recursive 
feasibility and some stability properties. 

The stochastic tube-based approaches, for example, make a wide use of 
probabilistic invariant or reachable sets to pose deterministic constraints in 
the nominal prediction such that chance constraints are satisfied, see 
\cite{cannon2011stochastic,hewing2018stochastic}. Also in 
\cite{cannon2009probabilistic}, probabilistic invariant sets are employed to 
handle probabilistic state constraints and a method for computing probabilistic 
invariant ellipsoids is presented.

Concerning the computation of reachable and invariant sets for deterministic 
systems and for robust control, i.e. in the worst-case disturbance context, 
several well-established results are present in the literature, for linear 
\cite{blanchini2008set,kolmanovsky1998theory} and nonlinear systems 
\cite{fiacchini2010computation}. In the recent years, some results have been 
appearing also on probabilistic reachable and invariant sets. The work 
\cite{kofman2012probabilistic} is completely devoted to the problem of computing 
probabilistic invariant sets and ultimately bounds for linear systems affected 
by additive stochastic disturbances. Also the paper 
\cite{hewing2018correspondence} presents a characterization of probabilistic 
sets based on the invariance property in the robust context, whereas 
\cite{hewing2019scenario} employ scenario-based methods to design them. 

In most of the works concerning probabilistic reachable and invariant sets 
computation and SMPC, however, the stochastic disturbance is modelled by an 
independent sequence of random variables. The assumption of independence, and 
thus uncorrelation, in time between disturbances, though, is often unrealistic. 
In this paper, we consider the problem of characterizing and computing, via 
convex optimization, outer bounds of probabilistic reachable sets and 
probabilistic invariant ellipsoids for linear systems excited by disturbances 
whose realizations are correlated in time. Only bounds on covariance and 
correlation matrices are required to be known, even stationarity is not 
necessary. Based on these bounds, the called correlation bound is defined and 
then employed to determine constructive conditions for computing probabilistic 
reachable and invariance ellipsoidal sets. The method, resulting in convex 
optimization problems, is then illustrated through numerical examples.

\emph{Notation:} 
The set of integers and natural numbers are denoted with $\Z$ and $\N$, 
respectively. The spectral radius of $A \in \R^{n \times n}$ is $\rho(A)$. 
The set of symmetric matrices in $\R^{n \times n}$ is denoted $\S^n$. With 
$\Gamma\succ 0$  ($S\succeq 0$) it is denoted that $\Gamma$ is a definite 
(semi-definite) positive matrix. If $\Gamma \succeq 0$ then 
$\Gamma^{\frac{1}{2}}$ is the matrix satisfying $\Gamma^{\frac{1}{2}} 
\Gamma^{\frac{1}{2}} = \Gamma$. For all $\Gamma \succeq 0$ and $r \geq 0$ 
define $\B(\Gamma, r) = \{x = \Gamma^{1/2} z \in \R^n: \, z\T z \leq r\}$; 
if moreover $\Gamma \succ 0$, then $\B(\Gamma, r) = \{x \in \R^n: \, x\T 
\Gamma^{-1} x \leq r\}$. Given two sets $Y, Z \subseteq \R^n$, their 
Minkowski set addition is $Y + Z = \{y + z \in \R^n : \ y \in Y, \, z \in Z\}$, 
their difference is $Y - Z = \{x \in \R^n : \ x + Z \subseteq Y\}$. The Gaussian 
(or normal) distribution with mean $\mu$ and covariance $\Sigma$ is denoted 
$\mathcal{N}(\mu, \Sigma)$, the $\chi$ squared cumulative distribution function 
of order $n$ is denoted $\chi_n^2(x)$.

\section{Correlation bound}

Consider the discrete-time system 
\begin{equation}\label{eq:system}
x_{k+1} = Ax_k + w_k,
\end{equation}
where $x_k \in \R^n$ is the state and $w_k \in \R^n$ an additive disturbance 
given by a sequence of random variables that are supposed to be correlated in 
time. 

\begin{remark}
In this paper, no assumption on $\{w_k\}_{k \in \N}$ is posed other than the 
existence of a bound on the covariance and correlation matrices. Neither 
stationarity is required. This aspect might be crucial in practice, as no exact 
knowledge of the matrices nor guarantee of stationarity are often available.
\end{remark}

The following definition of correlation bound encloses the key concept that 
permits to characterize and compute probabilistic reachable and invariant sets 
for linear systems affected by correlated disturbance. 

\begin{definition}[Correlation bound]\label{def:correlation:bound}
The random sequence $\{w_k\}_{k \in \Z}$ is said to have a correlation bound 
$\Gamma_w$ for matrix $A$ if the recursion $z_{k+1} = A z_k + w_k$ with $z_0 = 
0$, satisfies 
\begin{equation}\label{eq:corr_bound}
A\E{z_{k} w_k\T } + \E{w_k z_k\T}A\T + \E{w_k w_k\T } \preceq \Gamma_w, 
\end{equation}
or, equivalently 
\begin{equation*}\label{eq:corr_bound2}
\E{z_{k+1} z_{k+1}\T} \preceq A \E{z_k z_k\T}A\T + \Gamma_w, 
\end{equation*}
for all $k \geq 0$.
\end{definition}

\subsection{Computation of a correlation bound}

As it will be shown in the subsequent sections, a correlation bound permits
to determine sequences of probabilistic reachable sets and probabilistic 
invariant sets. For this, it is necessary to provide a condition and a method to 
obtain a correlation bound. Such a condition is presented in the following 
proposition.

\begin{proposition}\label{prop:LMIbound}
Given the system (\ref{eq:system}) with  $\rho(A) < 1$, let $\{w_k\}_{k \in \Z} 
\in \R^n$ be a random sequence such that 
\begin{equation}\label{eq:bound}
 \Gij \Gma^{-1}  \Gij^\top \preceq (\alpha + \beta \gamma^{j - i}) \Gma, 
\qquad \forall i \leq j,
\end{equation}
where 
\begin{align}
& \E{w_k w_k^\top} = \Gamma_{k,k} \preceq \Gma, \qquad  \forall k \in \N, 
\label{eq:Gamma}\\
& \Gij = \E{w_i w_j^\top}, \qquad \forall i,j \in \N,
\end{align}
with $\Gma \succ 0$ and $\alpha, \beta, \gamma \in \R$ such that $\alpha 
\geq 0$, $\beta \geq 0$ and $\gamma \in (0, \, 1)$. Given $\eta 
\in [\rho(A)^2, \, 1)$ and $\varphi \geq 1 $ such that there is $S \in \S^n$ 
satisfying
\begin{align}\label{eq:ASphi}
    & S \preceq \Gma \preceq \varphi S, \qquad A S A^\top \preceq \eta S
\end{align}
and $p \in (\eta, \, 1)$, then $\Gamma_w \in \S^n$ with $\Gamma_w 
\succ 0$ is a correlation bound for the sequence $\{w_k\}_{k \in \Z}$ and  
matrix $A$ if it satisfies
\begin{equation}\label{eq:LMIcond}
\Big(\alpha \varphi \frac{\eta}{p - \eta} + \beta \varphi \frac{\gamma \eta}{p - 
\gamma \eta}   + \frac{p}{1- p} + 1 
\Big) \Gma \preceq \Gamma_w.
\end{equation}

\end{proposition}

\vspace{0.15cm}

\begin{proof}
From the definition of correlation bound and the equality $z_k = 
\sum_{i=0}^{k-1} A^{k-1-i}w_i$, matrix $\Gamma_w$ must satisfy 
\begin{align}
A \E{(\sum_{i = 0}^{k-1}A^{k-1-i} w_i)w_k^\top } + \E {w_k(\sum_{i = 
0}^{k-1}A^{k-1-i} w_i)^\top } A^\top \nonumber \\ 
+ \E{w_k w_k^\top} \preceq \Gamma_w \nonumber 
\end{align}
for all $k \in \N$. From condition (\ref{eq:bound}) and 
\begin{align}
& 0 \preceq \left(\frac{A^{j-i} \Gij \Gma^{ -\frac{1}{2}}}{p^{\frac{j-i}{2}}} 
 - p^{\frac{j-i}{2}} \Gma^{\frac{1}{2}} \right) \hspace{-0.1cm} 
\left(\frac{A^{j-i} \Gij \Gma^{-\frac{1}{2}}}{p^{\frac{j-i}{2}}}  - 
p^{\frac{j-i}{2}} \Gma^{\frac{1}{2}} \right)^\top \nonumber\\
& \hspace{-0.15cm} = \hspace{-0.05cm} p\hspace{-0.1cm}^{-(j-i)} 
\hspace{-0.05cm} 
A^{j-i} \Gij \Gma^{-1} \hspace{-0.05cm} \Gij^\top (A^{j-i} 
)\hspace{-0.1cm}^\top \hspace{-0.15cm} + 
\hspace{-0.05cm} p^{j-i}\Gma 
\hspace{-0.1cm} - \hspace{-0.05cm} A^{j-i} \Gij \hspace{-0.05cm} - 
\hspace{-0.05cm} \Gij^\top (A^{j-i})\hspace{-0.1cm}^\top \nonumber
\end{align}
for every $i,j \in \N$ with $i \leq j$ and $p \neq 0$, it follows that 
\begin{align}
& A^{j-i} \Gij + \Gij^\top (A^{j-i})^\top \nonumber \\
& \preceq (\alpha p^{-(j-i)} + \beta (\gamma p^{-1})^{j-i}) A^{j-i} \Gma 
(A^{j-i})^\top + p^{j-i} \Gma.  \nonumber
\end{align}
Therefore, for every $k \in \N$ it holds 
\begin{align}\label{eq:LMI2}
& \hspace{-0.1cm} A \E{(\sum_{i = 0}^{k-1}A^{k-1-i} w_i)w_k^\top } +  \E 
{w_k(\sum_{i = 0}^{k-1}A^{k-1-i} w_i)^\top } A^\top \nonumber\\
& + \E{w_k w_k^\top} \hspace{-0.1cm} \preceq \hspace{-0.1cm} \sum_{i = 
0}^{k-1}A^{k-i} \E{w_i w_k^\top} + \hspace{-0.1cm} \sum_{i = 0}^{k-1} \E{w_k 
w_i^\top}  (A^{k-i})^\top \hspace{-0.1cm}  + \Gma  \nonumber\\
& = \Big( \sum_{i = 0}^{k-1}A^{k-i} \Gik +  \Gik^\top  (A^{k-i})^\top \Big) + 
\Gma \nonumber\\
& \preceq \Big( \sum_{i = 0}^{k-1} (\alpha p^{-(k-i)} \hspace{-0.1cm} + 
\hspace{-0.05cm} \beta (\gamma p^{-1})^{k-i}) A^{k-i} \Gma (A^{k-i})^\top 
\hspace{-0.1cm} + \hspace{-0.05cm} p^{k-i} \Gma  \Big) \hspace{-0.1cm}  + 
\hspace{-0.05cm} \Gma. \nonumber
\end{align}
From (\ref{eq:ASphi}), it follows that
\begin{equation}\label{eq:SGamma}
A^j \Gma (A^j)^\top \preceq \varphi A^j S (A^j)^\top \preceq \varphi \eta^{j} S 
\preceq \eta^{j} \Gma
\end{equation}
for all $j \in \N$, and then 
\begin{align}\label{eq:LMI3}
\hspace{-0.1cm} &  A \E{(\sum_{i = 0}^{k-1} \hspace{-0.1cm} A^{k-1-i} 
w_i)w_k^\top \hspace{-0.05cm}} \hspace{-0.05cm} + \hspace{-0.05cm} 
\E{w_k(\sum_{i = 0}^{k-1} \hspace{-0.1cm} A^{k-1-i} w_i)\hspace{-0.05cm}^\top 
\hspace{-0.05cm}} A\hspace{-0.05cm}^\top \hspace{-0.15cm} + \hspace{-0.05cm} 
\E{w_k w_k^\top\hspace{-0.05cm}}  
\nonumber \\
\hspace{-0.1cm} & \preceq \sum_{i = 0}^{k-1} \alpha \varphi(\eta p^{-1} 
)^{k-i}\Gma \hspace{-0.1cm}+ \hspace{-0.1cm}\sum_{i = 0}^{k-1} 
\beta \varphi(\gamma \eta p^{-1})^{k-i}\Gma \hspace{-0.1cm} + 
\hspace{-0.1cm}\sum_{i = 0}^{k-1} p^{k-i} \Gma + \Gma \nonumber\\
& = \Big( \sum_{j = 1}^{k} \alpha \varphi(\eta p^{-1})^j + \sum_{j = 
1}^{k} \beta \varphi (\gamma \eta p^{-1})^j + \sum_{j = 1}^{k} 
p^j\Big) \Gma + \Gma \nonumber\\
& = \Big( \alpha \varphi(\eta p^{-1} ) \frac{1 - (\eta p^{-1} )^k}{1 
- \eta p^{-1} }  + \beta \varphi (\gamma \eta p^{-1}) \frac{1- 
(\gamma \eta p^{-1})^k}{1- \gamma \eta p^{-1}} \nonumber\\
& \hspace{0.25cm} + p \frac{1- p^k}{1 - p}\Big) \Gma + \Gma. 
\end{align}
Two possibilities exist, $\eta$ can be either positive or zero. If $\eta > 0$ 
then $0 < \gamma \eta < \eta < p < 1$, and all the terms in the summation in 
(\ref{eq:LMI3}) are positive and monotonically increasing with $k$. If $\eta = 
0$ the first two terms in (\ref{eq:LMI3}) are null and the third one, i.e. 
$p(1-p^k)/(1-p)$, is positive and monotonically 
increasing with $k$, since $0 = \eta < p < 1$. In both cases the supremum is 
finite and attained for $k \rightarrow +\infty$ and then condition 
(\ref{eq:LMIcond}) implies that $\Gamma_w$ is a correlation bound for $A$.
\end{proof}

Note that condition (\ref{eq:bound}) is the reasonable assumption of a 
correlation that exponentially vanishes with time. For one dimensional systems 
and $\alpha = 0$, for instance, it means that the correlation function of $w_i$ 
and $w_j$ is exponentially vanishing as $|j - i|$ grows. 
\begin{remark}
Notice moreover that only an upper bound on the covariance $\Gma$, ensuring the 
satisfaction of (\ref{eq:bound}), is necessary to be known. This is also 
reasonable, since the exact values of $\Gij$ for all $i,j \in \N$ are often not 
available, in practice.
\end{remark}

The result of Proposition~\ref{prop:LMIbound} is used hereafter to design 
an optimization based procedure to compute the tightest correlation bound.
To obtain the sharper bound through (\ref{eq:LMIcond}), the parameter 
multiplying $\Gma$, has to be minimized. Note first that such parameter is 
monotonically increasing with $\varphi$ and $\eta$, for $\varphi \geq 1$ 
and $\eta \in [\rho(A)^2,1)$. Nevertheless, the minimizing pair $\varphi$ and 
$\eta$ is not evident, even for a given $p$, due to the constraint 
(\ref{eq:ASphi}). One possibility is to grid the interval $[\rho(A)^2, 1)$ of 
$\eta$ and then obtain, for every value of $\eta$ on the grid, the optimal 
$\varphi$ and $p$. To do so, one should first fix $\eta$ and then solve the 
semidefinite programming problem 
\begin{align}
(\varphi^*, \, S^*)  = & \min_{\varphi, S} \ \varphi \nonumber \\
& \mathrm{s.t. } \quad \Gma \preceq S \preceq \varphi \Gma \nonumber \\
& \hspace{0.7cm} A S A\T \preceq \eta S. \nonumber
\end{align}
 
Note now that, once $\eta$ and $\varphi$ are given, condition 
(\ref{eq:LMIcond}) is a convex constraint in $p$ and then in $\Gamma_w$. In 
fact, $a/(p - a)$ is zero if $a = 0$ and it is finite, convex and 
decreasing for $p \in (a, +\infty)$ if $a > 0$, whereas $p/(1-p)$ is finite, 
convex and increasing for $p \in (-\infty, 1)$. Then, the minimum of the 
function multiplying $\Gma$ exists and is unique in $(\eta,1)$. This means that, 
once $\varphi$ and $\eta$ are fixed, the value of $p$ that minimizes the 
parameter multiplying $\Gma$ at the lefthand-side of (\ref{eq:LMIcond}) can be 
computed by solving the following convex optimization problem in a scalar 
variable:
\begin{align}
p^*(\eta, \varphi)   = & \min_{p} \ \alpha \varphi \frac{\eta}{p - 
\eta} + \beta \varphi \frac{\gamma \eta}{p - \gamma \eta}   + \frac{p}{1- 
p} \nonumber \\
& \mathrm{s.t. } \quad \eta < p < 1. \nonumber
\end{align}
Finally, $\Gamma_w$ can be computed by using in (\ref{eq:LMIcond}) the minimal 
value of the parameter multiplying $\Gma$ over the optimal ones obtained for 
the different $\eta$ on the grid and then minimizing a measure of $\Gamma_w$, 
or, even, get $\Gamma_w$ by imposing the equality to hold in (\ref{eq:LMIcond}).

\begin{remark}
Note that $\gamma$ could also be bigger than or equal to~1: this would lead 
to an (although non realistic) increasingly correlated disturbance. The limit 
would exist provided that $\eta$ is smaller than the inverse of $\gamma$, for 
all $p \in (\gamma \eta, 1)$. The case of $\gamma = 1$ is realistic, for 
instance for the case of constant disturbances, and can modelled by the 
constant term $\alpha$.
\end{remark}

The dependence of the bound (\ref{eq:LMIcond}) on the parameter $\varphi$ can 
be removed by avoiding using the bound $S \preceq \varphi \Gma$ as in 
(\ref{eq:SGamma}). The corollary below, providing a potentially less 
conservative correlation bound, follows straightforwardly. 

\begin{corollary}\label{cor:LMIbound}
Under the hypothesis of Proposition~\ref{prop:LMIbound}, given $p \in 
(\eta,1)$, $\Gamma_w$ is a correlation bound for matrix $A$ if it satisfies
\begin{equation}\label{eq:LMIcond_S}
\Big(\frac{\alpha \eta}{p - \eta} + \frac{\beta \gamma 
\eta}{p - \gamma \eta}\Big) S + \Big(\frac{p}{1- p} + 1 
\Big) \Gma \preceq \Gamma_w.
\end{equation}
\end{corollary}

\vspace{0.15cm}

Condition (\ref{eq:LMIcond_S}) provides a further degree of freedom, i.e. the 
matrix $S$, that can be used to improve the bound.

\section{Probabilistic reachable and invariant sets}

Based on the correlation bound, conditions for computing probabilistic reachable 
and invariant sets are presented. First, two properties are given that are 
functional to the purpose. 

\begin{property}\label{prop:inclusion}
For every $r > 0$ and every $\Gma, \Sigma \in \S^n$ such that $\Gma 
\succeq 0$ and $\Sigma \succ 0$, it holds
\begin{equation}\label{eq:inclusion}
\B(A \Gma A \T + \Sigma, r) \subseteq  A \B(\Gma, r) + 
\B(\Sigma,r).
\end{equation}
\end{property}

\begin{proof}
Notice first that $A \Gma A \T + \Sigma \succ 0$ and then
\begin{align}
\B(A \Gma A \T + \Sigma, r) & = \{x \in \R^n: \ x^\top (A \Gma A \T + 
\Sigma)^{-1} x \leq r\} \nonumber \\
A \B(\Gma, r) + \B(\Sigma,r) & = \{x = A \Gma^{1/2} y + 
\Sigma^{1/2} w \in \R^n: \nonumber \\ 
& \hspace{0.5cm} y^\top y \leq r, \ w^\top w \leq r\}. 
\label{eq:LMIconditions2} 
\end{align}
For a given $x \in \B(A \Gma A \T + \Sigma, r)$, the vectors $y$ and 
$w$ defined 
\begin{equation}\label{eq:yw}
y = \Gma^{1/2} A^\top (A \Gma A^\top + \Sigma)^{-1} x, \quad w = 
\Sigma^{1/2} (A \Gma A^\top + \Sigma)^{-1} x 
\end{equation}
are such that 
\begin{equation*}
A \Gma^{1/2} y \hspace{-0.05cm} + \hspace{-0.05cm} \Sigma^{1/2} w 
\hspace{-0.05cm} = \hspace{-0.05cm} A \Gma A^\top \hspace{-0.1cm} (A \Gma 
A^\top \hspace{-0.1cm}  + \Sigma)^{-1} x \hspace{-0.05cm} + \hspace{-0.05cm} 
\Sigma (A \Gma A^\top \hspace{-0.1cm} + \Sigma)^{-1} x = x. 
\end{equation*}
Moreover, 
\begin{align}
y\T y = x\T (A \Gma A^\top + \Sigma)^{-1} A \Gma A\T (A \Gma A^\top + 
\Sigma)^{-1} x \nonumber \\
\leq x\T (A \Gma A^\top + \Sigma)^{-1} x \leq r \nonumber
\end{align}
since $A \Gma A\T \preceq A \Gma A^\top + \Sigma$ and $x \in \B(A 
\Gma A \T + \Sigma, r)$. Analogously
\begin{align}
w\T w = x\T (A \Gma A^\top + \Sigma)^{-1} \Sigma (A \Gma A^\top + 
\Sigma)^{-1} x \nonumber \\
\leq x\T (A \Gma A^\top + \Sigma)^{-1} x \leq r \nonumber
\end{align}
from $\Sigma\preceq A \Gma A^\top + \Sigma$. Hence, given $x \in \B(A 
\Gma A \T + \Sigma, r)$, two vectors $y$ and $w$ exist, as defined in 
(\ref{eq:yw}), such that $x = A \Gma^{1/2} y + \Sigma^{1/2} w$ and $y\T y 
\leq r$ and $w\T w \leq r$, which means that $x \in  A \B(\Gma, r) 
+ \B(\Sigma, r)$, from (\ref{eq:LMIconditions2}). Thus 
(\ref{eq:inclusion}) is proven. 
\end{proof}

The result in Property~\ref{prop:inclusion} is used in the following one, to 
characterize bounds on the covariance matrices and probabilities of the system 
trajectory.

\begin{property}\label{prop:sevaral_results}
Suppose that the random sequence $\{w_k\}_{k \in \N}$ has a correlation bound 
$\Gamma_w \succ 0$ for matrix $A$ with $\rho(A) < 1$. Given $r > 0$, consider 
the system $z_{k+1} = Az_k + w_k$ with $z_0 = 0$ and the recursion
\begin{align}
& \Gamma_{k+1} = A \Gamma_k A^\top + \Gamma_w \label{eq:Gammak}
\end{align}
with $\Gamma_0 = 0 \in \R^{n \times n}$. Then,
\begin{itemize}
 \item[(i)] $\E{z_k z_k^\top} \preceq \Gamma_k, \quad \forall k \geq 0$, 
\item[(ii)] $\Pr\{z_k \in \B(\Gamma_k , r)\} \geq 1 - \frac{n}{r}, \quad \forall 
k \geq 1$,
\item[(iii)] $\B(\Gamma_k , r) \subseteq \B(\Gamma_{k+1}, r) \subseteq 
A \B(\Gamma_k, r) + \B(\Gamma_w , r), \quad \forall k \geq 1$.
\end{itemize}
\end{property}

\vspace{0.2cm}

\begin{proof} The claims are proved successively.
\begin{itemize}
 \item[(i)] Suppose that $\E{z_k z_k^\top} \preceq \Gamma_k$ with $\Gamma_k$ 
recursively defined through (\ref{eq:Gammak}). Then
\begin{equation*}
\begin{array}{l}
\hspace{-0.6cm} \E{z_{k+1} z_{k+1}^\top} = \E{A z_k z_k^\top A^\top 
\hspace{-0.2cm} + A z_k w_k^\top \hspace{-0.2cm} + w_k z_k^\top A^\top 
\hspace{-0.2cm} + w_k w_k^\top }\\
\hspace{-0.4cm} = A \E{z_k z_k^\top } A^\top \hspace{-0.2cm} + A\E{z_k 
w_k^\top} 
+ \E{w_k z_k^\top }A^\top \hspace{-0.2cm} + \E{w_k w_k^\top }\\
\hspace{-0.4cm} \preceq A \E{z_k z_k^\top } A^\top \hspace{-0.2cm} + \Gamma_w 
\preceq A \Gamma_k A^\top \hspace{-0.2cm} + \Gamma_w = \Gamma_{k+1},
\end{array}
\end{equation*}
where the first inequality follows from the definition of correlation bound.

 \item[(ii)] This result is based on the Chebyshev inequality, 
\cite{stellato2017multivariate,navarro2016very}. From Markov's 
inequality, \cite{bertsekas2008introduction,billingsley2008probability}, a 
nonnegative random variable $x$ with expected value $\mu$, 
satisfies $\Pr\{x > r \} \leq \mu/r$ for all $r > 0$. From $\Gamma_w 
\succ 0$, it follows that $\Gamma_k \succ 0$ and 
$\Gamma_k^{-1} \succ 0$ for all $k \geq 1$ and then there exists $D_k \in \R^{n 
\times n}$ such that $\Gamma_k^{-1} = D_k^\top D_k$ for all $k \geq 1$. Thus
\begin{equation*}
\begin{array}{l}
\E{z_{k}^\top \Gamma_k^{-1} z_{k} } = \E{z_{k}^\top D_k^\top D_k z_{k}} = 
\E{\tr\{z_{k}^\top D_k^\top D_k z_{k}\}} \\ 
\hspace{0.2cm} = \E{\tr\{D_k z_{k}z_{k}^\top D_k^\top \}} = \tr\{D_k 
\E{z_{k}z_{k}^\top \}D_k^\top} \\
\hspace{0.2cm} \leq \tr\{D_k \Gamma_{k}D_k^\top\} = \tr\{\Gamma_{k}D_k^\top 
D_k\} = \tr\{I\} = n
\end{array}
\end{equation*}
and then, by applying the Markov's inequality, one gets $\Pr\{z_k^\top 
\Gamma_k^{-1} z_k > r\} \leq n/r$ and hence $\Pr\{z_k^\top 
\Gamma_k^{-1} z_k \leq r\} \geq 1 - n/r$, for all $k \geq 1$.

\item[(iii)] 
From the definition of $\Gamma_k$, it follows $ \Gamma_{k} = \sum_{i = 0}^{k-1} 
A^i \Gamma_w (A^i) ^\top$ for $k \geq 1$ and then
\begin{align}
\hspace{-0.1cm} \Gamma_{k+1} \hspace{-0.05cm} = \hspace{-0.05cm} A^{k} \Gamma_w 
(A^{k})^\top \hspace{-0.1cm} + \hspace{-0.1cm} \sum_{i = 0}^{k-1} A^i \Gamma_w 
(A^i) ^\top \hspace{-0.1cm} = \hspace{-0.05cm} A^{k} \Gamma_w (A^{k})^\top 
\hspace{-0.1cm} + \hspace{-0.05cm} \Gamma_k \succeq \Gamma_k. \nonumber 
\end{align}
This implies $\Gamma_{k+1}^{-1} \preceq \Gamma_{k}^{-1}$ and hence, $\B(\Gamma_{k},r) \subseteq \B(\Gamma_{k+1},r)$ for all $k \geq 1$. The inclusion  $\B(\Gamma_{k+1}, r) \subseteq A \B(\Gamma_k, r) + \B(\Gamma_w , r)$ follows by applying Property~\ref{prop:inclusion} with the definition of 
$\Gamma_{k+1}$ as in (\ref{eq:Gammak}).
\end{itemize}
\vspace*{-0.5cm}
\end{proof}

\subsection{Probabilistic reachable sets}

The simplest confidence regions are ellipsoids, that have been widely used in 
the context of MPC, see, for example, 
\cite{cannon2011stochastic,hewing2018stochastic}. The definition of 
probabilistic reachable sets is recalled.

\begin{definition}[Probabilistic reachable set]
It is said that $\Omega_k \subseteq \R^n$ with $k \in \N$ is a sequence of 
probabilistic reachable sets for system (\ref{eq:system}), with violation 
level $\epsilon \in [0,1]$, if $x_0 \in \Omega_0$ implies $\Pr\{ x_k \in 
\Omega_k\} \geq 1 - \epsilon$ for all $k \geq 1$.
\end{definition}

A condition for a sequence of sets to be a probabilistic reachable sets is 
presented, in terms of correlation bound. The analogous result for uncorrelated 
disturbance is in \cite{hewing2018correspondence}.

\begin{proposition}\label{prop:reachable}
Suppose that the random sequence $\{w_k\}_{k \in \N}$ has a correlation bound 
$\Gamma_w \succ 0$ for matrix $A$ with $\rho(A) < 1$. Given $r > 0$, consider 
the system (\ref{eq:system}) and the recursion (\ref{eq:Gammak}) with $x_0 = 0 
\in \R^n$, $\Gamma_0 = 0 \in \R^{n \times n}$. Then the sets defined as 
\begin{equation}\label{eq:cRk}
\cR_{k+1} = A \cR_k + \B(\Gamma_w,r), 
\end{equation}
for all $k \in \N$, and $\cR_0 = \{0\}$ are probabilistic reachable sets with 
violation level $n/r$ for every $r > 0$. 
\end{proposition}

\begin{proof} 
It will be firstly proved that $\B(\Gamma_k,r) \subseteq \cR_k$, for all 
$k\geq 1$. Note first that $\Gamma_1  =  A\Gamma_0A\T +\Gamma_w= \Gamma_w$ 
and $\cR_1  = A \cR_0 + \B(\Gamma_w,r)$. Thus, $\B(\Gamma_1,r) = 
\B(\Gamma_w,r)=\cR_1$ and hence the claim is satisfied for $k=1$. It suffice 
now to prove that $\B(\Gamma_k,r)\subseteq \cR_k$ implies 
$\B(\Gamma_{k+1},r)\subseteq \cR_{k+1}$. Supposing 
$\B(\Gamma_k,r) 
\subseteq \cR_k$ implies 
\begin{equation*}
\B(\Gamma_{k+1},r) \hspace{-0.05cm} \subseteq \hspace{-0.05cm} A \B(\Gamma_k,r) 
\hspace{-0.05cm} + \B(\Gamma_w,r) \hspace{-0.05cm} \subseteq \hspace{-0.05cm} 
A\cR_k \hspace{-0.05cm} + \hspace{-0.05cm} \B(\Gamma_w,r) \hspace{-0.05cm} = 
\hspace{-0.05cm} \cR_{k+1}, 
\end{equation*}
where the first inclusion follows from (iii) of 
Property~\ref{prop:sevaral_results}. From this and the second claim of 
Property~\ref{prop:sevaral_results}, it follows 
\begin{equation*}
\Pr\{x_k\in \cR_k\} \geq \Pr\{ x_k \in \B(\Gamma_k,r)\} \geq 1 - 
\frac{n}{r}, 
\end{equation*}
which implies that $\cR_k$ with $k \in \N$ is a sequence of probability 
reachable sets with violation level $n/r$.
\end{proof}

\subsection{Probabilistic invariant sets}

The concept of probabilistic invariant sets, as defined and used in 
\cite{kofman2012probabilistic,hewing2018correspondence}, is recalled.

\begin{definition}[Probabilistic invariant set]
The set $\Omega \subseteq \R^n$ is a probabilistic invariant set for the system 
(\ref{eq:system}), with violation level $\epsilon \in [0,1]$, if $x_0 \in 
\Omega$ implies $\Pr\{ x_k\in \Omega\} \geq 1-\epsilon$ for all $k \geq 1$.
\end{definition}

A first condition for a set to be probabilistic invariant, analogous to that 
proved in \cite{hewing2018correspondence} for uncorrelated disturbances, is 
given 
below. 

\begin{property}\label{prop:Inv}
Suppose that the random sequence $\{w_k\}_{k \in \N}$ has a correlation bound 
$\Gamma_w \succ 0$ for matrix $A$. If $W \in \S^n$ and $r > 0$ are such that 
$W \succ 0$ and 
\begin{equation}\label{eq:PIS_LMI_TEO}
A \B(W, 1) + \B(\Gamma_w,r) \subseteq \B(W,1),
\end{equation}
then $\B(W,1)$ is a probabilistic invariant set with violation 
probability $n/r$.
\end{property}

\begin{proof} By definition, it is sufficient to show that $x_0 \in 
\B(W,1)$ implies $\Pr\{x_k\in \B(W,1)\} \geq 1 - n/r$, for all $k 
\geq 0$. The state $x_k$ can be written as the sum of a nominal term $\bx_k$ 
and a random vector $z_k$ that depends on the past realizations of the 
uncertainty. That is, $x_k = \bx_k+z_k$, where $\{\bx_k\}_{k\geq 0}$ and 
$\{z_k\}_{k\geq 0}$ are given by the recursions 
\begin{equation}\label{eq:xz}
\bx_{k+1} = A\bx_k, \qquad  z_{k+1} = A z_k + w_k,
\end{equation}
for all $k \geq 0$, with $\bx_0 = x_0$ and $z_0 = 0$. Below it 
is first proved that $x_0\in \B(W,1)$ and (\ref{eq:PIS_LMI_TEO}) imply 
\begin{equation}\label{inc:bx:Rk:P}
\bx_k + \cR_k \subseteq \B(W,1), \ \forall k\geq 0,
\end{equation} 
with $\cR_k$ as in (\ref{eq:cRk}). Since $\cR_0 = \{0\}$, the inclusion is 
trivially satisfied for $k=0$. Supposing that $\bx_{k}+ \cR_{k} \subseteq 
\B(W,1)$ yields 
\begin{align}
\bx_{k+1} + \cR_{k+1} & = A\bx_{k} + \left( A \cR_{k} + 
\B(\Gamma_w,r)\right) \nonumber \\ 
& = A \left(\bx_{k}+ \cR_{k}\right) + \B(\Gamma_w,r) \nonumber \\
& \subseteq A \B(W,1) + \B(\Gamma_w,r) \subseteq \B(W,1), 
\nonumber 
\end{align} 
and then (\ref{inc:bx:Rk:P}) holds. Condition (\ref{inc:bx:Rk:P}) implies 
\begin{equation*}
\Pr\{x_k \hspace{-0.05cm} \in \B(W \hspace{-0.05cm}, 1)\} \hspace{-0.05cm} 
= \hspace{-0.05cm} \Pr \{ \bx_k +z_k \hspace{-0.05cm} \in \hspace{-0.05cm} 
\B(W \hspace{-0.05cm}, 1) \} \hspace{-0.05cm} \geq \hspace{-0.05cm} 
\Pr\{z_k \hspace{-0.05cm} \in \hspace{-0.05cm} \cR_k\} \hspace{-0.05cm} \geq 
\hspace{-0.05cm} 1 \hspace{-0.05cm} - \hspace{-0.05cm} \frac{n}{r}
\end{equation*}
for all $k \geq 0$, where the last inequality follows from 
Proposition~\ref{prop:reachable}.
\end{proof}

Property \ref{prop:Inv} implies that the existence of a correlation bound 
provides a condition for probabilistic invariance that has the same structure as 
the one corresponding to robust invariance. In case of ellipsoidal invariant, 
(\ref{eq:PIS_LMI_TEO}) results in a bilinear condition, see 
\cite{boyd1994linear}, that can be solved, for instance, by gridding the space 
of the Lagrange multiplier and solving an LMI for every value. 

An additional novel condition for a set to be probabilistic invariant, 
employing the correlation bound of the correlated random sequence $w_k$, 
follows. 

\begin{proposition}\label{prop:PIScondition}
Suppose that the random sequence $\{w_k\}_{k \in \Z}$ has a correlation bound 
$\Gamma_w \succ 0$ for matrix $A$. If $W \in \S^n$ and $\lambda \in [0,1)$ are 
such that $W \succ 0$ and 
\begin{equation}\label{eq:PIS_LMI}
A W A^\top + \Gamma_w \preceq W
\end{equation}
and 
\begin{equation}\label{eq:PIS_LMIb}
A W A^\top \preceq \lambda W
\end{equation}
then $\B(W,\rho)$ is a probabilistic invariant set with violation probability 
$n/(1-\lambda)\rho$. If, moreover, $w_k$ has normal 
distribution, then $\B(W,\rho)$ is a probabilistic invariant set with violation 
probability $1-\chi_n^2((1-\lambda)\rho)$.
\end{proposition}

\begin{proof} 
As in the proof of Property~\ref{prop:Inv}, denote $x_k = \bx_k+z_k$, where 
$\{\bx_k\}_{k\geq 0}$ and $\{z_k\}_{k\geq 0}$ are given by the recursions 
(\ref{eq:xz}) for all $k \geq 0$, with $\bx_0 = x_0$ and $z_0 = 0$. Then, from 
$x_0 \in \B(W,\rho)$ and (\ref{eq:PIS_LMIb}), it follows that $\bx_k \in 
\B(W,\lambda^k \rho)$ for all $k \in \N$. First it is proved that $\E{z_k 
z_k\T} \preceq W$ implies $\E{z_{k+1} z_{k+1}\T} \preceq W$ for every $k \in 
\N$. Since $z_0 = 0$, the inequality $\E{z_0 z_0\T} = 0 \preceq W$ is trivially 
satisfied. Suppose now that $\E{z_k z_k\T} \preceq W$, then 
\begin{align}
 & \E{z_{k+1} z_{k+1}\T}  = \E{(A z_k + w_k) (A z_k + w_k)^\top } \nonumber \\
& = \E{A z_k z_k\T A\T + A z_k w_k\T \hspace{-0.15cm} + w_k z_k\T A\T 
\hspace{-0.15cm}+ w_k w_k\T } 
\nonumber \\
 & = A\E{z_k z_k\T} A\T \hspace{-0.15cm} + A \E{z_k w_k\T } + \E{w_k z_k\T } 
A\T + \E{w_k w_k\T} \nonumber \\
& \preceq A W A\T + \Gamma_w \preceq W \nonumber 
\end{align}
where the first inequality follows from the definition of correlation bound 
and the second from (\ref{eq:PIS_LMI}). Note now that
\begin{align}
& \Pr\{ x_k \in \B(W,\rho)\} = \Pr\{\bx_k + z_k \in \B(W,\rho)\} \nonumber \\
& \geq \Pr\{z_k \in \B(W,\rho) - \B(W,\lambda^k \rho)\} \nonumber 
\end{align}
since $\bx_k \in \B(W,\lambda^k \rho)$ with probability 1. It follows that 
$\B(W,(1-\lambda)\rho) = \B(W,\rho) - \B(W, \lambda \rho) \subseteq 
\B(W,\rho)  - \B(W, \lambda^k \rho)$ and, from and $\E{z_k z_k\T} \preceq W$ and 
the Chebyshev inequality (see, for example, proof of claim (ii) of 
Property~\ref{prop:sevaral_results}), then 
\begin{align}
& \Pr\{ x_k \in \B(W,\rho)\} \geq \Pr\{z_k \in \B(W,(1-\lambda)\rho)\} \nonumber 
\\
& = \Pr\{z_k W^{-1} z_k \leq (1-\lambda)\rho\} \geq 1 - 
\frac{n}{(1-\lambda)\rho}. \nonumber 
\end{align}
The results for $w_k$ with normal distribution follow directly from the 
definition of the $\chi$ squared cumulative distribution, that is $\Pr\{y\T y 
\leq r\} = \chi_n^2(r)$ for $y \in \R^n$ with standard normal 
distribution and $r > 0$, see 
\cite{billingsley2008probability,bertsekas2008introduction}.
\end{proof}

The proposition below proves that the convex condition (\ref{eq:PIS_LMI}) is 
less conservative than (\ref{eq:PIS_LMI_TEO}). 

\begin{proposition}\label{prop:LMIconditions}
Suppose that the random sequence $\{w_k\}_{k \in \N}$ has a correlation bound 
$\Gamma_w \succ 0$ for matrix $A$. If $W$ is such that condition 
(\ref{eq:PIS_LMI_TEO}) holds for $r \geq 1$, then also (\ref{eq:PIS_LMI}) is 
satisfied.
\end{proposition}

\begin{proof}
First note that condition (\ref{eq:PIS_LMI}) is equivalent to $\B(A W A \T 
\hspace{-0.15cm} + \Gamma_w,1) \subseteq \B(W,1)$. From claim (iii) of 
Property \ref{prop:LMIbound} it follows 
\begin{equation*}
\B(A W A \T \hspace{-0.15cm} + \Gamma_w, \hspace{-0.05cm} 1) 
\hspace{-0.05cm} \subseteq \hspace{-0.05cm} A \B(W\hspace{-0.05cm}, 1) 
\hspace{-0.05cm} + \hspace{-0.05cm} \B(\Gamma_w,1) \hspace{-0.05cm} 
\subseteq \hspace{-0.05cm} A \B(W, 1) \hspace{-0.05cm} + 
\hspace{-0.05cm} \B(\Gamma_w,r)
\end{equation*}
since $r \geq 1$, and thus that (\ref{eq:PIS_LMI_TEO}) implies 
(\ref{eq:PIS_LMI}). 
\end{proof}

Therefore, condition (\ref{eq:PIS_LMI}) can be used to efficiently determine 
probabilistic invariant ellipsoids and is less conservative than 
(\ref{eq:PIS_LMI_TEO}). 

\section{Numerical examples}

Two examples are presented, concerning different bounds on the correlation 
matrices of the random sequence~$w_k$ that affects the system 
(\ref{eq:system}) with 
$A = \left[\begin{smallmatrix}
    0.25 & 0\\  0.1 &  0.3
     \end{smallmatrix}\right].
$

\subsection{Exponentially decaying bound}

Consider the case in which the bound (\ref{eq:bound}) holds with $\alpha = 0$, 
$\beta = 1$ and $\gamma < 1$. This means that the correlation between the 
samples of $w_k$ and $w_{k-l}$ is assumed to be exponentially decreasing with 
$l$ and it represents the systems for which the dependence between samples 
fades with time. 

To validate the presented results, it is necessary to generate a random sequence 
satisfying the bound (\ref{eq:bound}). This can be obtained by feeding an 
asymptotically stable discrete-time system with an i.i.d. random process with 
zero mean and a given constant covariance matrix. Consider the i.i.d. process 
$u_k$ with $k \in \N$ such that  
\begin{equation*}
 \E{u_k} = 0\ \qquad \E{u_k u_k\T} = U,
\end{equation*}
for all $k \in \N$ and 
\begin{equation}\label{eq:sysHF}
 w_{k+1} = H w_k + F u_k.
\end{equation}
Then $w_{k} = \sum_{j = -\infty}^{k-1} H^{j-k+1} F u_{j}$ and it can be proved 
that 
\begin{equation*}
 \E{w_k} = 0, \quad \E{w_k w_k\T} = \Gma, \quad \E{w_{k+l} w_k\T} = H^l 
\Gma
\end{equation*}
for all $k, l \in \N$, where $\Gma \in \S^n$ is the unique solution of $\Gma 
= H \Gma H\T + F U F\T$. Then, the bound (\ref{eq:bound}) holds with $\alpha = 
0$, $\beta = 1$ and $\gamma$ solution of the following optimization problem:
\begin{align}
& \min_{\gamma} \{ \gamma \ : \  H \Gma H\T \preceq \gamma 
\Gma \}.
\end{align}
Therefore, given $\Gma$ and $\gamma$, a random sequence satisfying 
(\ref{eq:bound}) 
for this values can be obtained by appropriately designing $H$ and $F$ or, 
viceversa, given $U$, $H$ and $F$, the matrix $\Gma$ and $\gamma$ can be 
obtained.

An i.i.d. random sequence with distribution $\mathcal{N}(0,U)$, with $U = 
\rm{diag}(1.5, \, 0.26)$, has been used to feed system (\ref{eq:sysHF}) 
with 
\begin{equation*}
H = \left[\begin{array}{cc}
     0.75 & -0.2\\ 0 & 0.6 
     \end{array}\right], \qquad 
F  = \left[\begin{array}{cc}
     1 & 2\\ 0.5 & -3 
     \end{array}\right]
\end{equation*}
giving a correlated sequence $w_k$ with null mean and covariance matrix 
\begin{equation*}
\Gma = \left[\begin{array}{cc}
   7.8381 &  -2.3983\\
   -2.3983 &   4.2422
\end{array}\right].
\end{equation*}
The bound (\ref{eq:bound}) holds with $\alpha = 0$, $\beta = 1$ and $\gamma = 
0.676$. The correlation bound $\Gamma_w$, computed using (\ref{eq:LMIcond_S}), 
and matrix $W$ from (\ref{eq:PIS_LMI}) are
\begin{equation*}
\Gamma_w = \left[\hspace{-0.1cm}
\begin{array}{cc}
   19.5198 &  -5.9726\\
   -5.9726 &  10.5646
\end{array}\hspace{-0.1cm}\right]\hspace{-0.1cm}, \
W = \left[\hspace{-0.1cm}\begin{array}{cc}
   20.8211  & -5.8942\\
   -5.8942  & 11.4496
\end{array}\hspace{-0.1cm}\right]
\end{equation*}
with $\lambda = 0.1221$ from (\ref{eq:PIS_LMIb}). Using the fact that $w_k$ 
has normal distribution, the set $\B(W,\rho)$ is a probabilistic invariant set 
with violation probability of $p_v$ with $\chi_2^2((1-\lambda)\rho) = 1 - p_v$. 
Different values of violation probability $p_v$ have 
been tested, in particular $p_v = 0.1, 0.2, 0.3, 0.4, 0.5$. For every $p_v$, $N 
= 1000$ initial states $x_0$ have been uniformly generated on the boundary of 
$\B(W,\rho)$ and assumed independent on $w_k$. For each $x_0$, a sequence $w_k$ 
has been generated and applied. For every $k = 1, \ldots, 100$, the set of 
states $x_k$ are computed and the number of violation $v_k$ of the constraint 
$x_k \in \B(W,\rho)$ have been computed. The frequencies of violation $v_k/N$, 
for every $p_v$ and $k \in 1,\ldots,100$, are depicted in Fig.~\ref{fig:1}, 
that shows that the bound is always satisfied. 

\begin{figure}[!ht]
\hspace*{-0.5cm}
\input{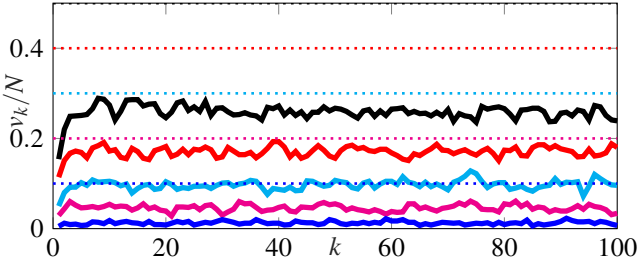}
\vspace*{-0.4cm}
     \caption{Frequency of violations $v_k/N$ of $x_k \in \B(W,\rho)$ for $k = 
1, \ldots,100$, with $\alpha = 0$ and $\beta = 1$, obtained for violation 
probability of: $50\%$ in black; $40\%$ in red; $30\%$ in cyan; $20\%$ in 
magenta; $10\%$ in blue.}\label{fig:1}
\end{figure}

\subsection{Constant bound}

A case for which the bound (\ref{eq:bound}) holds with $\alpha = 1$ and $\beta 
= 0$ is considered here. Supposing the that $w_k = w$ for all $k \in \N$, such 
that $\E{w w\T} \preceq \Gma$, would lead (\ref{eq:bound}) to hold with 
$\alpha = 1$ and $\beta = 0$.

The constant value of the disturbance $w$ has been generated according to 
$\mathcal N(0,\Gma)$ with $\Gma$ randomly generated:
\begin{equation*}
\Gma = \left[\begin{array}{cc}
     0.4785  & -0.7254\\
   -0.7254   & 1.5215
\end{array}\right].
\end{equation*}
The correlation bound as in (\ref{eq:LMIcond_S}) and related probabilistic 
invariant are given by the matrices 
\begin{equation*}
\Gamma_w = \left[\hspace{-0cm}
\begin{array}{cc}
    1.1877 &  -1.8007\\
   -1.8007 &   3.7767
\end{array}\hspace{-0cm}\right]\hspace{-0.1cm}, 
\ W  = \left[\hspace{-0cm}
\begin{array}{cc}
    1.2669 &  -1.9125\\
   -1.9125 &   4.0380
\end{array}\hspace{-0cm}\right].
\end{equation*}
with $\lambda =  0.0921$ from (\ref{eq:PIS_LMIb}). As for the case of decaying 
bound, the violation probabilities $p_v = 0.1, 0.2, 0.3, 0.4, 0.5$ are 
considered and $1000$ initial states $x_0$ are uniformly distributed on the 
boundary of $\B(W,\rho)$ to check the violation frequencies. For every 
$x_0$, a constant sequence $w_k = w$, with $k \in \N$, is generated with 
distribution $\mathcal N(0,\Gma)$ for $w$ and the number of the set inclusion 
$x_k \in \B(W,\rho)$ violations $v_k$ are obtained for all $k = 1, \ldots,100$. 
See the results in Fig~\ref{fig:2}, the probability violation bound is 
satisfied.

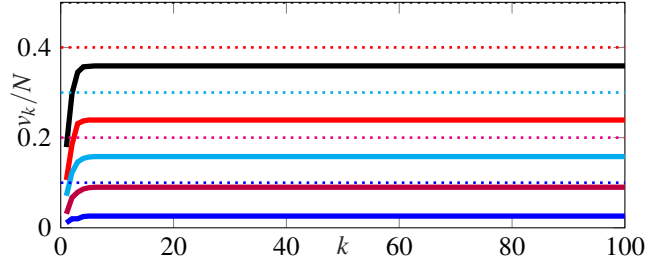
\begin{figure}[!ht]
\hspace*{-0.5cm}
%
%
%
\begin{tikzpicture}

\begin{axis}[%
width=7.5cm,
height=3cm,
at={(0,0)},
scale only axis,
xmin=0,
xmax=100,
ymin=0,
ymax=0.5,
ylabel style={xshift=.2cm, yshift=-0.5cm, font=\color{white!15!black}},
ylabel={$v_k/N$},
xlabel style={yshift=0.5cm, font=\color{white!15!black}},
xlabel={$k$},
axis background/.style={fill=white}
]
\addplot [color=black, line width=2.0pt, forget plot]
  table[row sep=crcr]{%
1	0.179000000000002\\
2	0.298000000000002\\
3	0.344999999999999\\
4	0.356999999999999\\
6	0.358999999999995\\
100	0.358999999999995\\
};
\addplot [color=red, line width=2.0pt, forget plot]
  table[row sep=crcr]{%
1	0.105999999999995\\
2	0.185000000000002\\
3	0.230999999999995\\
4	0.236999999999995\\
5	0.239000000000004\\
100	0.239000000000004\\
};
\addplot [color=cyan, line width=2.0pt, forget plot]
  table[row sep=crcr]{%
1	0.070999999999998\\
2	0.123000000000005\\
3	0.146000000000001\\
4	0.153999999999996\\
5	0.156999999999996\\
6	0.158000000000001\\
100	0.158000000000001\\
};
\addplot [color=purple, line width=2.0pt, forget plot]
  table[row sep=crcr]{%
1	0.0310000000000059\\
2	0.0669999999999931\\
3	0.0789999999999935\\
4	0.0859999999999985\\
5	0.0889999999999986\\
6	0.0900000000000034\\
100	0.0900000000000034\\
};
\addplot [color=blue, line width=2.0pt, forget plot]
  table[row sep=crcr]{%
1	0.0109999999999957\\
2	0.019999999999996\\
3	0.019999999999996\\
4	0.0250000000000057\\
5	0.0259999999999962\\
100	0.0259999999999962\\
};
\addplot [color=black, dotted, line width=1.0pt, forget plot]
  table[row sep=crcr]{%
0	0.5\\
100	0.5\\
};
\addplot [color=red, dotted, line width=1.0pt, forget plot]
  table[row sep=crcr]{%
0	0.400000000000006\\
100	0.400000000000006\\
};
\addplot [color=cyan, dotted, line width=1.0pt, forget 
plot]
  table[row sep=crcr]{%
0	0.299999999999997\\
100	0.299999999999997\\
};
\addplot [color=magenta, dotted, line width=1.0pt, forget plot]
  table[row sep=crcr]{%
0	0.200000000000003\\
100	0.200000000000003\\
};
\addplot [color=blue, dotted, line width=1.0pt, forget plot]
  table[row sep=crcr]{%
0	0.0999999999999943\\
100	0.0999999999999943\\
};
\end{axis}
\end{tikzpicture}%
\vspace*{-0.4cm}
     \caption{Frequency of violations $v_k/N$ of $x_k \in \B(W,\rho)$ for $k = 
1, \ldots,100$, with $\alpha = 1$ and $\beta = 0$, obtained for violation 
probability of: $50\%$ in black; $40\%$ 
in red; $30\%$ in cyan; $20\%$ in magenta; $10\%$ in blue.}\label{fig:2}
\end{figure}

\section{Conclusions}

This paper presented methods, based on convex optimization, to compute 
probabilistic reachable and invariant sets for linear systems fed by a 
stochastic disturbance correlated in time. From the knowledge of some bound on 
the correlation matrices, the characterization of the called correlation 
bound is given and then employed for obtaining the reachable and invariant 
sets.

\bibliographystyle{plain}
\bibliography{autosam_arXiv}

\end{document}